\documentclass[aps,pra,10pt,twocolumn,showpacs, superscriptaddress, nobalancelastpage, nofootinbib]{revtex4-1} %
\usepackage[utf8]{inputenc}

\usepackage{graphicx} 	

\usepackage{amsthm}
\usepackage{amssymb}

\usepackage{bm}		
\usepackage{mathtools}  
\usepackage{dsfont}	
\usepackage{wrapfig}
\usepackage{multirow}
\usepackage{dcolumn} 	

\usepackage{hyperref}

\newtheorem{theorem}{Theorem}

\newtheorem{observation}[theorem]{Observation}

\DeclareMathOperator{\tr}{tr}
\DeclareMathOperator{\Tr}{Tr}
\DeclareMathOperator{\wt}{wt}
\DeclareMathOperator{\supp}{supp}

\DeclareMathOperator{\SWAP}{SWAP}

\newcommand{\ot}[0]{\otimes}
\newcommand{\nn}[0]{\nonumber}

\newcommand{\nhf}[0]{\lfloor \frac{n}{2} \rfloor}

\newcommand{\mc}[1]{\mathcal{#1}}

\renewcommand{\AA}[0]{\mathcal{A}}
\newcommand{\BB}[0]{\mathcal{B}}

\newcommand{\EE}[0]{\mathcal{E}}

\newcommand{\QQ}[0]{\mathcal{Q}}
\newcommand{\CC}[0]{\mathcal{C}}
\newcommand{\N}{\mathds{N}}

\newcommand{\C}{\mathds{C}}

\newcommand{\one}[0]{\mathds{1}}

\newcommand{\bra}[1]{\langle #1|}
\newcommand{\ket}[1]{|#1\rangle}

\newcommand{\ketbra}[2]{| #1\rangle \langle #2|}
\newcommand{\dyad}[1]{| #1\rangle \langle #1|}  

\renewcommand{\a}{\alpha}
\renewcommand{\b}{\beta}
\renewcommand{\r}{\varrho}

\begin{document}

\title			{Bounds on absolutely maximally entangled states from shadow inequalities, \\ and the quantum MacWilliams identity}
\date		{\today}

\author		{Felix Huber}
\affiliation	{Naturwissenschaftlich-Technische Fakult\"at, Universit\"at Siegen, D-57068 Siegen, Germany}

\author	 	{Christopher Eltschka}
\affiliation	{Institut für Theoretische Physik, Universit\"at Regensburg, D-93040 Regensburg, Germany}

\author		{Jens Siewert}        
\affiliation	{Departamento de Qu\'{i}mica F\'{i}sica, Universidad del Pa\'{i}s Vasco UPV/EHU, E-48080 Bilbao, Spain}
\affiliation	{IKERBASQUE Basque Foundation for Science, E-48013 Bilbao, Spain}

\author		{Otfried G\"uhne}
\affiliation	{Naturwissenschaftlich-Technische Fakult\"at, Universit\"at Siegen, D-57068 Siegen, Germany}

\begin{abstract}
  A pure multipartite quantum state is called absolutely maximally entangled (AME), 
  if all reductions obtained by tracing out at least half of its parties
  are maximally mixed. Maximal entanglement is then present across every bipartition.
  The existence of such states is in many cases unclear. 
  With the help of the weight enumerator machinery known from quantum
  error correction and the generalized shadow inequalities, we obtain new bounds
  on the existence of AME states in dimensions larger than two. To complete the treatment on the 
  weight enumerator machinery, the quantum MacWilliams identity is derived in the Bloch representation.
  Finally, we consider AME states whose subsystems have different local dimensions, 
  and present an example for a \(2\times3\times3\times3\) system that shows 
  maximal entanglement across every bipartition.
\end{abstract}

\maketitle


\section{Introduction}

\setlength\intextsep{0pt}
\setlength{\columnsep}{1pt}%
\begin{wrapfigure}{l}{0.07\textwidth}
    \includegraphics[width=0.07\textwidth]{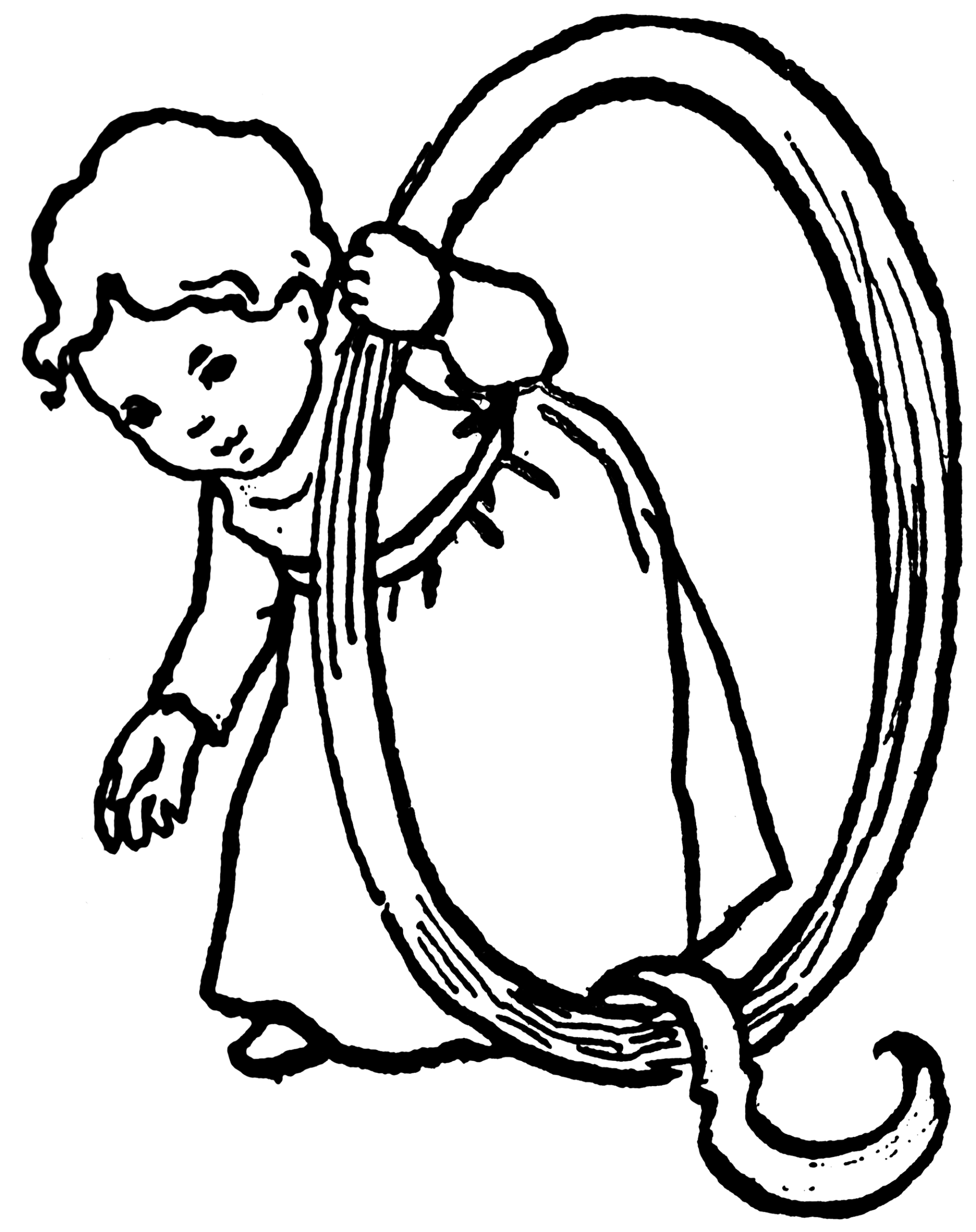}
\end{wrapfigure}
\setlength\intextsep{0pt}
\noindent 
uantum states of many particles show interesting non-classical features, 
foremost the one of entanglement.
A pure state of \(n\) parties is called {\em absolutely maximally entangled} (AME), 
if all reductions to \(\nhf\) parties are maximally mixed. 
Here, \(\lfloor \cdot \rfloor\) is the floor function. 
Then maximal possible entanglement is present across each bipartition. 
Well-known examples are the Bell and GHZ states on two and three parties respectively.
AME states have been shown to be a resource for a variety of quantum information-theoretic tasks 
that require maximal entanglement amongst many parties, such as open-destination teleportation, 
entanglement swapping, and quantum secret sharing~\cite{Helwig2012, Helwig2013}. 
They also represent building blocks for holographic quantum error-correcting codes, 
and are often called perfect tensors in this context~\cite{Latorre2015, Pastawski2015, Li2017}.
Thus, it is a natural question to ask for what number of parties and local dimensions such states 
may exist~\cite{Gisin1998, Higuchi2000, Scott2004}.

The existence of AME states composed of two-level systems was recently solved:
Qubit AME states do only exist for $n=2,3,5,$ and $6$ parties, all of which can be 
expressed as graph or stabilizer states~\cite{Scott2004, Huber2017}. 
Concerning larger local dimensions however, the existence of such states is only partially resolved.
AME states exist for any number of parties, if the dimension of the subsystems is chosen large enough~\cite{Helwig2013}.
Furthermore, different constructions for such states have been put forward, based on graph states~\cite{Hein2006, Helwig_graph}, 
classical maximum distance separable codes~\cite{Grassl2004, Helwig2013}, 
and combinatorial designs~\cite{Goyeneche2015, Goyeneche2017}. 
However, for many cases it is still unknown whether or not AME states exist~\footnote{
For the current status of this question, see Problem~\(35\) in the list of Open
  Quantum Problems, IQOQI Vienna (November 2017),
  \href{http://oqp.iqoqi.univie.ac.at/existence-of-absolutely-maximally-entangled-pure-states}
  {http://oqp.iqoqi.univie.ac.at/existence-of-absolutely-maximally-entangled-pure-states}.}.

In this article, we give results on the question of AME state existence 
when the local dimension is three or higher. Namely, we show that, additionally
to the known non-existence bounds, 
three-level AME states of \(n = 8,12,13,14,16,17,19,21,23\),
four-level AME states of \(n = 12,16,20,24,25,26,28,29,30,33,37,39\), and 
five-level AME states of \(n = 28,32,36,40,44,48\) parties~do~not~exist.

To this end, we make use of the weight enumerator machinery known 
from quantum error correcting codes (QECC). With it, bounds can also be obtained
for one-dimensional codes, which are pure quantum states~\cite{Scott2004}.
We will make use of the so-called shadow inequalities, which constrain the admissible correlations of multipartite states, 
to exclude the existence of the above-mentioned AME states. 
Along the way, we will prove a central theorem, the quantum MacWilliams identity, 
originally derived by Shor and Laflamme for qubits~\cite{Shor1997} 
and by Rains for arbitrary finite-dimensional systems in Ref.~\cite{Rains1998}.
Thus our aim is twofold:
On the one hand, we provide an accessible introduction into the weight enumerator machinery 
in terms of the Bloch representation, in order to gain physical intuition. 
On the other hand, we apply this machinery to exclude the existence of certain higher-dimensional AME states
by making use of the shadow inequalities.

This article is organized as follows. In the next section, we introduce the shadow inequalities,
from which we eventually obtain the bounds mentioned above. 
In Sec.~\ref{sec:bloch}, the Bloch representation of quantum states is introduced, 
followed by a short discussion of QECC
and their relation to AME states in Sec.~\ref{sec:qecc}. 
In Sec.~\ref{sec:shadow_enum}, we introduce the shadow enumerator,
the Shor-Laflamme enumerators are explained in Sec.~\ref{sec:weight_enum},
followed by the derivation of the quantum MacWilliams identity in Sec.~\ref{sec:MacWilliams}.
The shadow enumerator in terms of the Shor-Laflamme enumerator is derived in Sec.~\ref{sec:deriv_shadow_enum}, 
from which one can obtain bounds on the existence of QECC and of AME states in particular, 
which is presented in Sec.~\ref{sec:AME}. 
After considering AME states in mixed dimensions in Sec.~\ref{sec:AME_mix}, we conclude in Sec.~\ref{sec:conclusion}.


\section{Motivation}\label{sec:motivation}
Originally introduced by Shor and Laflamme~\cite{Shor1997}, 
Rains established the notion of weight enumerators in a series of landmark articles 
on quantum error correcting codes~\cite{Rains1998, Rains1999, Rains2000}. With it, he stated some of the 
strongest bounds known to date on the existence of QECC~\cite{Rains1999}.

In particular, in his paper on polynomial invariants of quantum codes~\cite{Rains2000}, 
Rains showed an interesting theorem, which proved to be crucial to obtain those bounds.
These are the so-called {\em generalized shadow inequalities}:
For all positive semi-definite Hermitian operators  \(M\) and \(N\) on parties 
\((1 \dots  n)\) and any fixed subset \(T \subseteq \{1 \dots n\}\), it holds that
\begin{equation}\label{eq:shadow_ineq}
	\sum_{S \subseteq \{1 \dots n\}} (-1)^{|S\cap T|} \Tr_S[\Tr_{S^c}(M) \Tr_{S^c}(N)] \geq 0 \,.
\end{equation}
Here and in what follows, \(S^c\) denotes the complement of subsystem \(S\) in \(\{1 \dots n\}\),
and the sum is performed over all possible subsets \(S\).
Note that if $M=N=\rho$ is a quantum state, the generalized shadow inequalities 
are consistency equations involving the purities of the marginals, 
i.e. they relate terms of the form \(\Tr[\Tr_{S^c}(\rho)^2]\), which in turn 
can be expressed in terms of linear entropies. 
Thus, these inequalities form an exponentially large set of monogamy relations for multipartite quantum states, 
applicable to any number of parties and local dimensions.

To state bounds on the existence of AME states of \(n\) parties having
local dimension \(D\) each, one could in principle just evaluate this expression by 
inserting the purities of AME state reductions.
However, in order to understand the connections to methods from quantum error correcting codes,
let us first recall the quantum weight enumerator machinery, 
including the so-called shadow enumerator, which is derived from Eq.~\eqref{eq:shadow_ineq}.
We will then rederive the central theorem, namely the quantum MacWilliams identity. 
Finally, we obtain new bounds for AME states with the help of the shadow inequalities.
In order to remain in a language close to physics, we will work exclusively in the Bloch representation.


\section{The Bloch representation}\label{sec:bloch}
Let us introduce the Bloch representation.
Denote by \(\{e_j\}\) an orthonormal basis for operators acting on \(\C^{D}\), 
such that \(\Tr(e_j^\dag e_k) = \delta_{jk} D\).
We require that \(\{e_j\}\) contains the identity (e.g. \(e_0 = \one\)),
and therefore all other basis elements are traceless (but not necessarily Hermitian).
Then, a local error-basis \(\mc{E}\) acting on \((\C^D)^{\ot n}\) can be formed by 
taking tensor products of elements in \(\{e_j\}\).
That is, each element \(E_\alpha \in\mc{E}\) can be written as
\begin{equation}
	E_\alpha = e_{\alpha_1} \ot \dots \ot e_{\alpha_n} \,.
\end{equation}
Because the single-party basis \(\{e_j\}\) is orthonormal, the relation 
\(\Tr(E_\alpha^\dag E_\beta) = \delta_{\alpha\beta} D^n\) follows.
For qubits, \(\mc{E}\) can be thought of to contain all
tensor products which can be built from the identity and the Pauli matrices; 
in higher dimensions, a tensor-product basis can be formed from elements of 
the Heisenberg-Weyl or the generalized Gell-Mann basis~\cite{Bertlmann2008}.
Further, denote by \(\supp(E)\) the support of operator \(E\), that is, 
the set of parties on which \(E\) acts non-trivially.
The weight of an operator is then size of its support, 
and we write \(\wt(E) = |\supp(E)|\).

Having defined a local error-basis \(\mc{E}\) acting on \((\C^D)^{\ot n}\), 
every operator on \(n\) systems that have \(D\) levels each
can in the Bloch representation be decomposed as
\begin{align}\label{eq:op_expansion}
	M &= \frac{1}{D^{n}} \sum_{E\in \EE}  \Tr(E^\dag M ) E \,.
\end{align}
As in the above decomposition, we will often omit the subindex \(\alpha\),
writing \(E\) for \(E_\alpha\). 
Also, most equations that follow contain sums over all elements \(E\) in \(\EE\), 
subject to constraints. In those cases we will often denote the constraints only below the 
summation symbol.

Given an operator \(M\) expanded as in Eq.~\eqref{eq:op_expansion}, 
its reduction onto a given subsystem \(S^c\) tensored by the identity on the complement \(S\) reads
\begin{equation}\label{eq:ptrace_op1}
	\Tr_{S}(M) \ot \one_{S} = D^{|S|-n} \!\!\!\! \sum_{\supp(E) \subseteq S^c} \Tr(E^\dag M) E \,.
\end{equation}
This follows from \( \Tr_{S} (E) = 0\) whenever \(\supp(E) \not \subseteq S^c\).
Interestingly, this can also be written as a quantum channel
whose Kraus operators form a unitary \(1\)-design~\cite{Gross2007}.
\begin{observation}\label{eq:ptrace_channel}
  The partial trace over subsystem \(S\) tensored by the identity on \(S\)
  can also be written as a channel,
  \begin{align}\label{eq:ptrace_op2}
    \Tr_{S}(M) \ot \one_{S} &= D^{-|S|}\sum_{\supp(E) \subseteq S} \!\!\!\! E M E^\dag \,.
  \end{align}
\end{observation}
The proof can be found in Appendix~A.


\section{Quantum Error Correcting Codes}\label{sec:qecc}
Let us introduce quantum error correcting codes and their relation to 
absolutely maximally entangled states.
A quantum error correcting code with the parameters \(((n, K, d))_D\) 
is a K-dimensional subspace \(\QQ\) of \((\C^D)^{\ot n}\), 
such that for any orthonormal basis 
\(\{\ket{i_\QQ}\}\) of \(\QQ\) and all errors \(E\in\mc{E}\) with \(\wt(E)<d\) 
\cite{Scott2004, Gottesmann2002},
\begin{equation} \label{eq:def_qecc}
	\bra{i_\QQ} E \ket{j_\QQ} = \delta_{ij} C(E)\,.
\end{equation}
Note that the constant \(C(E)\) only depends on the error~\(E\).
Above, \(d\) is called the distance of the code.
If \(C(E) = \Tr(E) / D^n\), the code is called pure. 
By convention, codes with \(K=1\) are only considered codes if they are pure.

From the definition follows that a one-dimensional code (also called self-dual), 
described by a projector \(\dyad{\psi}\), must fulfill \(\Tr(E \dyad{\psi}) = 0\)
for all \(E \neq \one \) of weight smaller than \(d\).
Thus, pure one-dimensional codes
of distance $d$ are pure quantum states whose reductions
onto \((d-1)\) parties are all maximally mixed.
AME states, whose reductions onto \(\nhf\) parties are maximally mixed, 
are QECC having the parameters \(((n,1, \nhf+1))_D\).


\section{The shadow enumerator}\label{sec:shadow_enum}
Let us introduce the shadow enumerator \(S_{MN}(x,y)\), and point out its usefulness.
Following Rains~\cite{Rains1998}, we first define
\begin{align}\label{eq:unitary_enum}
	\AA'_S (M,N) &= \Tr_S     [\Tr_{S^c}(M) \Tr_{S^c}(N)]  \,, \\
	\BB'_S (M,N) &= \Tr_{S^c} [\Tr_{S}(M)   \Tr_{S}(N)]\,.
\end{align}
Naturally, \(\AA'_S = \BB'_{S^c}\).
With this, we define
\begin{equation}\label{eq:shadow_coeffs}
	S_j(M,N) = \sum_{ |T|= j} 
	\sum_{S\subseteq \{1\dots n\}} (-1)^{|S \cap T^c|} \AA'_S(M,N) \,,
\end{equation}
where the sum is over all subsets \(T \subseteq \{1\dots n\}\) of size~\(j\).
Eq.~\eqref{eq:shadow_ineq} states that all $S_j$ must be non-negative.
Note however, that there is the term \(T^c\) instead of \(T\) in the exponent, 
compared to Eq.~\eqref{eq:shadow_ineq}, but this does not matter, as Eq.~\eqref{eq:shadow_ineq} holds for any \(T\).

The {\em shadow enumerator} then is the polynomial
\begin{equation}\label{eq:shadow_enum}
	S_{MN}(x,y) = \sum_{j=0}^n S_j(M,N) \, x^{n-j} y^j\,.
\end{equation}

Given a hypothetical QECC or an AME 
state in particular, its shadow enumerator must have non-negative coefficients. 
If this is not the case, one can infer that such a code or state cannot exist.
However, how do we obtain this enumerator?
Two paths come to mind:
First, if we are interested in a one-dimensional code (\(K=1\)), 
the purities of the reductions determine all \(\AA'_S(\QQ)\).
For AME states of local dimension \(D\), the situation is particularly simple: 
from the Schmidt decomposition, it can be seen that all 
reductions to $k$ parties must have the purity
\begin{equation}\label{eq:AME_purity}
	\Tr(\rho_{(k)}^2) = D^{- \min(k, n-k)}\,.
\end{equation}
Second, the coefficients of the so called Shor-Laflamme enumerator \(A_j(\QQ)\) 
may be known (see also below), from which the shadow enumerator can be obtained.

Generally, when dealing with codes whose existence is unknown,
putative weight enumerators can often be obtained by stating 
the relations that follow as a linear program (see Appendix~F)~\cite{Calderbank1998, Scott2004, Nebe2006}. 
If, for a set of parameters \(((n,K,d))_D\), 
no solution can be found, a corresponding QECC cannot exist.

In the following three sections, we aim to give a concise introduction as well as intuition to 
this enumerator theory.


\section{Shor-Laflamme enumerators} \label{sec:weight_enum}
In this section, we introduce the protagonists of the enumerator machinery, 
the {\em Shor-Laflamme (weight) enumerators}~\cite{Shor1997, Rains1998}. 
These are defined for any two given Hermitian 
operators \(M\) and \(N\) acting on \((\C^D)^{\ot n}\), and are invariants under local unitary operations.
Their (unnormalized) coefficients are given by~\footnote{
For dimensions larger than two, this definition is different, but equivalent, 
to the original definition as found in Ref.~\cite{Rains1998}.}
\begin{align}
	A_j(M,N) &= \sum_{\wt(E) = j} \Tr(E M) \Tr(E^\dag N) \,,\\
	B_j(M,N) &= \sum_{\wt(E) = j} \Tr(E M E^\dag N)        \,.
\end{align}
The corresponding enumerator polynomials are
\begin{align}
	A_{MN}(x,y) &= \sum_{j=0}^n A_j(M,N) x^{n-j} y^j  \label{eq:SL_enum1}\,, \\
	B_{MN}(x,y) &= \sum_{j=0}^n B_j(M,N) x^{n-j} y^j  \label{eq:SL_enum2}\,.
\end{align}
While it might not be obvious from the definition, these enumerators are independent of the local error-basis \(\EE\) chosen,
and are thus local unitary invariants. This follows from the fact that they can expressed as linear combinations
of terms having the form of Eq.~\eqref{eq:unitary_enum}. The exact relation will be made clear in Section~\ref{sec:MacWilliams}.

When dealing with weight enumerators, there is the following pattern, as seen above:
First define a set of coefficients [e.g, \(A_j(M,N)\)], from which
the associated polynomial, the enumerator, is constructed [e.g., \(A_{MN}(x,y)\)].
If \(M=N\), we will often write the first argument only, e.g. \(A_j(M)\), or leave it out alltogether.
In Table~\ref{table:enums}, we give an overview of the coefficients and enumerators
used in this article.

\begin{table*}[!ht]
\begin{tabular}{ l  l  l }
&  { Coefficient } & { Enumerator }\\
\noalign{\vskip 2mm}    
\hline
\noalign{\vskip 2mm}    
\multirow{2}{*}{Shor-Laflamme enum.:\quad\quad}
      & $A_j(M,N)  \,= \sum_{\wt(E) = j} \Tr(E M) \Tr(E^\dag N)$ \quad\quad& $A_{MN}(x,y) = \sum_{j=0}^n A_j(M,N) x^{n-j} y^j$ \\
      & $B_j(M,N)  \,= \sum_{\wt(E) = j} \Tr(E M E^\dag N)$      \quad\quad& $B_{MN}(x,y) = \sum_{j=0}^n B_j(M,N) x^{n-j} y^j$ \\
\noalign{\vskip 2mm}
\hline
\noalign{\vskip 2mm}    
  \multirow{4}{*}{Rain's unitary enum.:\quad\quad}
      &$\AA'_S(M,N)    = \Tr_S     [\Tr_{S^c}(M) \Tr_{S^c}(N)] $ &\\
      &$\BB'_S(M,N) \, = \Tr_{S^c} [\Tr_{S}(M)   \Tr_{S}(N)]   $ &\\
      \noalign{\vskip 1mm}    
      & $A'_j(M,N) \,= \sum_{|S|=j} \AA'_S(M,N)$ \quad\quad& $A'_{MN}(x,y) = \sum_{j=0}^n A'_j(M,N) x^{n-j} y^j$ \\
      & $B'_j(M,N) \,= \sum_{|S|=j} \BB'_S M,N)$   \quad\quad& $B'_{MN}(x,y) = \sum_{j=0}^n B'_j(M,N) x^{n-j} y^j$   \\
\noalign{\vskip 2mm}    
\hline
\noalign{\vskip 2mm}
       Shadow enumerator:\quad\quad       & $S_j(M,N) \,\,= \sum_{ |T|= j} \sum_S (-1)^{|S \cap T^c|} \AA'_S(M,N)$
       \quad\quad& $S_{MN}(x,y) \,= \sum_{j=0}^n S_j(M,N) x^{n-j} y^j$
\end{tabular}
\caption{An overview on the different weight enumerator polynomials and their coefficients, which are local unitary invariants.} 
\label{table:enums}
\end{table*}

Considering a QECC with parameters \(((n, K, d))_D\), 
one sets \(M=N\) to be equal to the projector \(\QQ\) onto the code space.
The following results concerning QECC and their Shor-Laflamme enumerators are known~\cite{Rains1998}:
The coefficients \(A_j = A_j(\QQ)\) and \(B_j = B_j(\QQ)\) are non-negative, and
\begin{align}
	K B_0 &= A_0 = K^2      \label{eq:KB0_A0}\,,\\
	K B_j &\geq A_j  	\label{eq:KBj_Aj}\,,
\end{align}
with equality in the second equation for \(j<d\).
In fact, these conditions are not only necessary but also sufficient 
for a projector \(\QQ\) to be a QECC (see Appendix~B.).
The distance of a code can thus be obtained in the following way:
if a projector \(\QQ\) fulfills the above conditions with equality for all \(j < d\),
then \(\QQ\) is a quantum code of distance \(d\)~\footnote{
See Theorems $2$ and $18$ in Ref.~\cite{Rains1998}, and Ref.~\cite{Shor1997}.}.
%
For pure codes, additionally \(A_j = B_j = 0\) for all \(1 < j < d\).
In particular, AME states have \(A_{j} = 0\) for all \(1< j < \nhf+1\); the
remaining $A_j$ can be obtained in an iterative way from Eq.~\eqref{eq:AME_purity} \cite{Scott2004, Huber2017}.

In the case of $\QQ=\dyad{\psi}$, the weight enumerators have a 
particularly simple interpretation: 
The coefficient $A_j$ measures the contribution to the purity of \(\dyad{\psi}\) by 
terms in \(\dyad{\psi}\) having weight \(j\) only,
while the dual enumerator measures the overlap of $\dyad{\psi}$ with itself, 
given an error-sphere of radius $j$. Furthermore, we have \(A_j = B_j\) for all \(j\), as a direct evaluation shows.

In the entanglement literature, \(A_j(\rho)\) is also called the correlation strength, or the two-norm of the 
\(j\)-body correlation tensor~\cite{Aschauer2004, Kloeckl2015}.
Concerning codes known as stabilizer codes, 
\(A_j\) and \(B_j\) count elements of weight $j$ 
in the stabilizer and in its normalizer respectively~\cite{Gottesmann1997}.

Let us now try to give some intuition for these enumerators for general Hermitian 
operators \(M\) and \(N\).
Note that the coefficients of the primary enumerator $A_j(M,N)$ form a 
decomposition of the inner product \(\Tr(MN)\).
This can be seen by writing \(M\) and \(N\) in the Bloch representation [Eq.~\eqref{eq:op_expansion}],
\begin{align}\label{eq:Tr(MN)}
  \Tr(MN) &= 
	     D^{-2n} \Tr\Big( \sum_{E}      \Tr(E M ) E^\dag
	  \, \sum_{E'}      \Tr(E'^\dag N ) E' \Big)  \nn\\
	  &= D^{-2n} \Tr\Big(\sum_{E}  \Tr(E M ) \Tr(E^\dag N ) \,E^\dag E  \Big) \nn\\
	  &= D^{-n} \sum_{j=0}^n A_j(M,N)  \,. 
\end{align}

On the other hand, the coefficients of the dual enumerator \(B_j(M,N)\) 
can be seen as a decomposition of \(\Tr(M) \Tr( N)\).
To see this, recall that by definition of the partial trace, 
\begin{equation}\label{eq:def_ptrace}
      \Tr_{S^c} [\Tr_{S}(M) \Tr_{S}(N)] = \Tr \big[ (\Tr_{S} (M) \ot \one_{S}) \,N \big] \,.
\end{equation}
As shown in Observation~\ref{eq:ptrace_channel}, the partial trace over parties in $S$ 
tensored by the identity on \(S\) can also be written as a quantum channel,
\begin{align}
  \Tr_{S}(M) \ot \one_{S} &= D^{-|S|}\sum_{\supp(E) \subseteq S} \!\!\!\! E M E^\dag \,. 
\end{align}
Thus \(B_j(M,N)\) decomposes \(\Tr(M) \Tr(N)\), 
\begin{align}\label{eq:Tr(M)Tr(N)}
    \Tr(M) \Tr(N) &= \Tr[ \Tr(M) \one\, N] \nn\\
		  &= D^{-n} \Tr ( \sum_{E}  E M E^\dag N) \, \nn \\
			  &= D^{-n} \sum_{j=0}^n \sum_{\wt(E)=j} \Tr (E M E^\dag N) \, \nn \\
			  &= D^{-n} \sum_{j=0}^n B_j(M,N) \,.
\end{align}
The insight gained from writing the partial trace in two different ways [c.f. Eqs.~\eqref{eq:ptrace_op1}, \eqref{eq:ptrace_op2}], 
and the decomposition 
of \(\Tr(MN)\) and \(\Tr(M) \Tr(N)\) in terms of the coefficients of the Shor-Laflamme enumerators [c.f. Eqs.~\eqref{eq:SL_enum1}, \eqref{eq:SL_enum2}]
will prove to be the essence of the MacWilliams identity, which we rederive in the following section.


\section{The quantum MacWilliams identity}\label{sec:MacWilliams}
In this section, we prove the quantum MacWilliams identity. 
It relates the two Shor-Laflamme enumerators \(A_{MN}(x,y)\) and \(B_{MN}(x,y)\) [Eqs.~\eqref{eq:SL_enum1}, \eqref{eq:SL_enum2}]
of arbitrary Hermitian operators \(M\) and \(N\).
\begin{theorem}[Quantum MacWilliams identity~\cite{Rains1998, Nebe2006}]
For any two Hermitian operators \(M\) and \(N\) acting on \(n\) systems
having \(D\) levels each, the following identity holds,
\begin{equation}
    A_{MN}(x,y) = B_{MN} \Big( \frac{x+(D^2-1)y}{D}, \frac{x-y}{D} \Big) \,.
\end{equation}
\end{theorem}
\begin{proof}
In order to prove this identity, 
one has to express the trace inner product of reductions in two different ways:
given the operator \(M\) expanded as in Eq.~\eqref{eq:op_expansion}, 
its reduction tensored by the identity reads [cf. Eq.~\eqref{eq:ptrace_op1}]
\begin{equation}
	\Tr_{S^c}(M) \ot \one_{S^c} = D^{|S^c|-n} \!\!\!\! \sum_{\supp(E) \subseteq S} \Tr(EM) E^\dag \,.
\end{equation}
Therefore,
\begin{align}\label{eq:ptrace_sum_A_warmup}
	& \Tr[\Tr_{S^c}(M) \ot \one_{S^c} \, N] \nn\\
	&= \Tr\Big( D^{|S^c|-2n} \!\!\!\! \sum_{\supp(E)\subseteq S} \!\!\!\! \Tr(E M) E^\dag \, \sum_{E'} \Tr(E'^\dag N) E' \Big) \nn\\
	&= D^{|S^c|-n} \!\!\!\! \sum_{\supp(E)\subseteq S} \!\!\!\! \Tr(E M)  \Tr(E^\dag N)\,.
\end{align}
Summing over all subsystems \(S\) of size $m$, one obtains
\begin{align}\label{eq:ptrace_sum_A}
	    & \sum_{|S|=m} \Tr[\Tr_{S^c}(M)\ot \one_{S^c} \,N] \nn \\
	    &= D^{|S^c|-n} \sum_{|S|=m} \,\, \sum_{\supp(E)\subseteq S} \!\!\!\! \Tr(E M)  \Tr(E^\dag N) \nn\\
	    &= D^{-m} \sum_{j=0}^m \binom{n}{m}\binom{m}{j} \binom{n}{j}^{-1} A_j(M,N) \nn\\
	    &= D^{-m} \sum_{j=0}^m \binom{n-j}{n-m} A_j(M,N) \,.
\end{align}
Above, the binomial factors account for multiple occurrences of terms having weight \(j\)
in the sum.
Note that Eq.~\eqref{eq:ptrace_sum_A} forms the coefficients of 
Rains' {\em unitary enumerator} [cf. \eqref{eq:unitary_enum}], defined as~\cite{Rains1998}
\begin{align}\label{eq:unitary_enum_coeffs}
      A'_m (M,N) &= \sum_{|S|=m} \AA_S'(M,N) \nn\\
		 &= \sum_{|S|=m} \Tr_S     [\Tr_{S^c}(M) \Tr_{S^c}(N)]  \,.
\end{align}
On the other hand, by expressing the partial trace as a quantum channel (see Obs.~\ref{eq:ptrace_channel})
and again summing over subsystems of size \(m\), we can write
\begin{align}\label{eq:ptrace_sum_B}
  &\sum_{|S|=m} \Tr[\Tr_{S} (M) \ot \one_{S} \,\, N] \nn\\
  &= \sum_{|S|=m} \Tr ( D^{-|S|} \!\!\!\!\sum_{ \supp(E) \subseteq S}  \!\!\!\! E M E^\dag N ) \nn\\
  &= D^{-m} \sum_{j=0}^m \binom{n}{m}\binom{m}{j} \binom{n}{j}^{-1} B_j(M,N) \nn\\
  &= D^{-m} \sum_{j=0}^m \binom{n-j}{n-m} B_j(M,N)\,.
\end{align}
Similar to above, Eq.~\eqref{eq:ptrace_sum_B} forms the coefficients of 
the unitary enumerator [cf. Eq.~\eqref{eq:unitary_enum}]
\begin{align}
    B'_m(M,N) &= \sum_{|S|=m} \BB'_{S}(M,N) \nn\\
	      &= \sum_{|S|=m} \Tr_{S^c} [\Tr_{S}(M)   \Tr_{S}(N)]\,.
\end{align}
Naturally, the corresponding unitary enumerator polynomials read
\begin{align}
  A'_{MN}(x,y) &= \sum_{j=0}^n A'_j(M,N) x^{n-j} y^j \,\\
  B'_{MN}(x,y) &= \sum_{j=0}^n B'_j(M,N) x^{n-j} y^j \,.
\end{align}

Using relations \eqref{eq:ptrace_sum_A} and \eqref{eq:ptrace_sum_B},
one can establish with the help of generating functions that
\begin{align}
	A'_{MN}(x,y) &= A_{MN}\Big(x+\frac{y}{D},\frac{y}{D}\Big)\,, \label{eq:unitary_enum_to_KL_enum_poly}\\
	B'_{MN}(x,y) &= B_{MN}\Big(x+\frac{y}{D},\frac{y}{D}\Big)\,. \label{eq:unitary_enum_to_KL_enum_poly_B}
\end{align}
This is somewhat tedious but straightforward (see Appendix~C).
It remains to use that \(\BB'_S(M,N) = \AA'_{S^c}(M,N)\), 
from which follows that \(B'_k(M,N) = A'_{n-k}(M,N)\), and
\begin{equation}
  A'_{MN}(x,y) = B'_{MN}(y,x)\,.
\end{equation}
Thus the quantum MacWilliams identity is established,
\begin{align}
	A_{MN}(x,y) &= A'_{MN}(x-y, Dy) 
	       = B'_{MN}(Dy, x-y) 			\nn\\
	       &= B_{MN}\Big(\frac{ x + (D^2-1)y}{D}, \frac{x-y}{D} \Big) \,.
\end{align}
This ends the proof.
\end{proof}

Because the relations~\eqref{eq:unitary_enum_to_KL_enum_poly} 
and \eqref{eq:unitary_enum_to_KL_enum_poly_B} are symmetric, one also has that 
\begin{equation}
	B_{MN}(x,y) 
	       = A_{MN}\Big(\frac{ x + (D^2-1)y}{D}, \frac{x-y}{D} \Big) \,.
\end{equation}
Thus the quantum MacWilliams transform is involutory.

Recall that for \(M=N= \dyad{\psi}\), one has \(A_j(\ket{\psi}) = B_j(\ket{\psi})\). 
Therefore the enumerator \(A_{\ket{\psi}}(x,y)\) must stay invariant under the transform
\begin{align}
      x \,\,\, &\longmapsto \,\,\,\frac{x + (D^2-1)y}{D} \,,\nn\\
      y \,\,\, &\longmapsto \,\,\,\quad\quad\frac{x-y}{D} \,.
\end{align}
In this case, a much simpler interpretation of the MacWilliams identity 
can be given: It ensures that the purities of complementary reductions, 
averaged over all complementary reductions of fixed sizes, are equal.

As shown above, the quantum MacWilliams identity is in essence a decomposition of the trace inner product of 
reductions of operators \(M\) and \(N\) in two different ways.
The motivation lies in the decomposition of \(\Tr(MN)\) and \(\Tr(M)\Tr(N)\),
using different ways to obtain the partial trace in the Bloch picture 
[cf. Eqs.~\eqref{eq:ptrace_op1} and Obs.~\ref{eq:ptrace_channel}].
Finally, note that the derivation of the identity did not require \(M,N\) to be positive semi-definite.
Therefore the quantum MacWilliams identity holds for all, including non-positive, 
pairs of Hermitian operators.


\section{The shadow enumerator in terms of the Shor-Laflamme enumerator}\label{sec:deriv_shadow_enum}
So far, we have introduced the Shor-Laflamme and the shadow enumerator.
Let us now see how to express one in terms of the other.
The strategy is the following: the shadow inequalities
are naturally expressed in terms of \(\AA'_S\) [cf. Eqs.~\eqref{eq:shadow_ineq} and \eqref{eq:shadow_coeffs}],
which we then write as a transformation of \(A_{MN}(x,y)\).

\begin{theorem}[Rains~\footnote{
  See Theorem $8$ in Ref.~\cite{Rains1999} and Theorem \(13.5.1.\) on p.~$383$ in Ref.~\cite{Nebe2006} for \(D=2\). 
  Also Sec. $V$ in Ref.~\cite{Rains1998} states this result, 
  but contains a sign error in the second argument of \(A_C\).}]
Given \(A_{MN}(x,y)\), the shadow enumerator is given by
\begin{equation}
	S_{MN}(x,y) = A_{MN}\left( \frac{ (D-1)x + (D+1)y}{D}, \frac{y-x}{D} \right)\,.
\end{equation}
\end{theorem}

\begin{proof}
Recall from Eq.~\eqref{eq:shadow_coeffs}, that for Hermitian operators 
\(M,N\geq 0\), the coefficients of the shadow enumerator are
\begin{equation}
	S_j(M,N) = \sum_{ |T|= j} 
	\sum_{S} (-1)^{|S \cap T^c|} \AA'_S(M,N) \,.
\end{equation}

As a first step, let us understand what combinatorial factor a given \(\AA'_S(M,N)\) 
receives from the sum over the subsets \(T \subseteq \{1 \dots n\}\) of size $j$,
or subsets $T^c$ of size \(m = n-j\) respectively.
For a fixed subsystem \(S\) of size \(k\), we can evaluate the partial sum
\begin{equation}\label{eq:minus_sum_over_partition_overlap}
    f(m=|T^c|, k=|S|;n) = \sum_{|T^c|=m} (-1)^{|S \cap T^c|} \,.
\end{equation}
By considering what possible subsets $T^c$ of size \(m\) have a 
constant overlap of size \(\alpha\) with $S$, yielding a sign \((-1)^\alpha\),
we obtain the expression
\begin{equation} 
    f(m, k;n) =  \sum_\alpha \binom{n-k}{m-\alpha} \binom{k}{\alpha} (-1)^{\alpha} \eqqcolon K_m(k;n)\,,
\end{equation}
where \(K_m(k;n)\) is the so-called Krawtchouk polynomial (see Appendix~D).
Above, $\binom{k}{\alpha}$ accounts for the different combinatorial possibilities 
of elements $T^c$ having overlap \(\alpha\) with $S$.
Necessarily, $T^c$ must then have a part of size $m-\alpha$ lying outside of $S$; 
there are $\binom{n-k}{m-\alpha}$ ways to obtain this. 
This is illustrated in Fig.~\ref{fig:Krawtchouk_combinatorics}.

\begin{figure}[!t]
\includegraphics[height=9em]{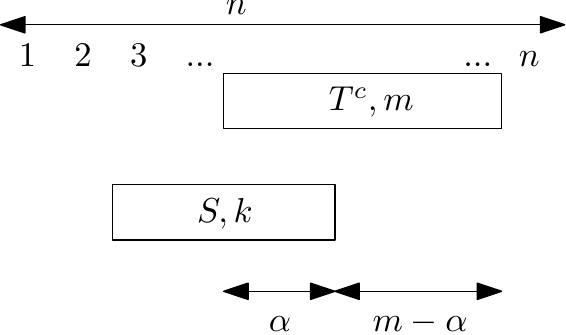}
 \caption{Overlap between $S$ and subsets $T$ of size m.
	  The term $\binom{k}{\alpha}$ accounts for the different combinatorial possibilities 
	  of elements $T^c$ having overlap \(\alpha\) with $S$.
	  Necessarily, $T^c$ must then have a part of size $m-\alpha$ lying outside of $S$; 
	  there are $\binom{n-k}{m-\alpha}$ ways to obtain this.}
\label{fig:Krawtchouk_combinatorics}
\end{figure}

Therefore, one obtains
\begin{equation}\label{eq:shadow_in_unitary_enum_Krawtchouk}
	S_j(M,N) = \sum_{k=0}^n K_{n-j}(k;n) A'_k(M,N) \,.
\end{equation}

Again, one can write this relation in a more compact form in 
terms of the unitary enumerator (see Appendix~E),
\begin{equation}\label{eq:shadow_in_unitary_enum}
	S_{M N}(x,y) = A_{MN}'(x+y,y-x)\,.
\end{equation}
To obtain the shadow enumerator in terms of the Shor-Laflamme enumerator, 
we take advantage of Eq.~\eqref{eq:unitary_enum_to_KL_enum_poly}. Then 
\begin{align}
	&S_{MN}(x,y) = A_{MN}'(x+y,y-x) \nn\\
	&= A_{MN}\left( \frac{ (D-1)x + (D+1)y}{D}, \frac{y-x}{D} \right)\,.
\end{align}
This ends the proof.
\end{proof}
Thus, given the Shor-Laflamme enumerator, one can obtain the shadow enumerator 
simply by a transform. If any of its coefficients are negative, a corresponding QECC cannot exist.

Given the parameters \( ((n,K,d))_D\) of a hypothetical QECC, one can formulate a linear program
to find possible enumerators which satisfy all the relations derived, 
namely Eq.~\eqref{eq:KB0_A0} and Eq.~\eqref{eq:KBj_Aj}, as well as the quantum MacWilliams identity (Thm.~2) 
and the quantum shadow identity (Thm.~3) (see Appendix~F)~\cite{Calderbank1998, Nebe2006}.
If no valid weights \(A_j\) can be found, a code with the proposed parameters cannot exist.
This provides a method to prove the non-existence of certain hypothetical states and QECC; 
on the other hand, the existence of a valid enumerator however does not imply the existence of a corresponding code.

An overview on the relations between the enumerators is given in Appendix~G.


\section{New bounds on absolutely maximally entangled states}\label{sec:AME}
\begin{figure*}[!th]
  \includegraphics[width=0.95\linewidth]{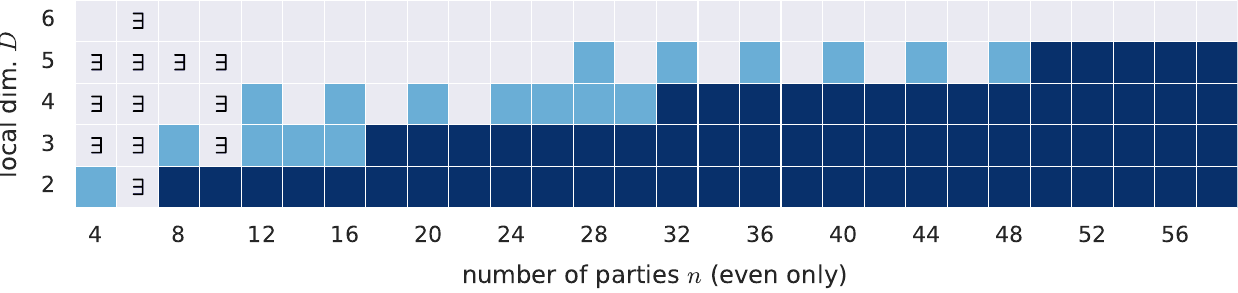}
\end{figure*}
\begin{figure*}[!th]
  \includegraphics[width=0.95\linewidth]{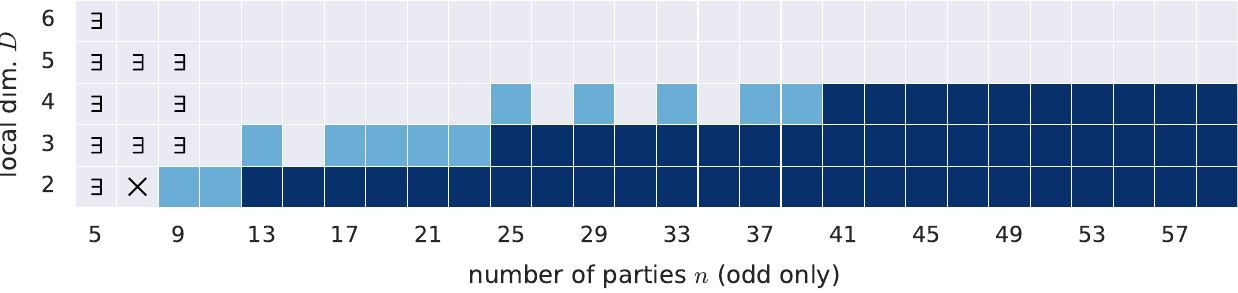}
  \caption{
	In dark blue, AME states are marked which are already excluded by the bound from Scott [Eq.~\eqref{eq:Scott_bound}];
	in light blue, those AME states are shown for which the negativity of the shadow enumerator coefficients \(S_j(\ket{\phi_{n,D}})\)
	[Eq.~\eqref{eq:AME_shadow_coeff}] gives stronger bounds. 
	The non-existence of AME states having parameters \(n=4,9,11\) with \(D=2\) was already known~\cite{Higuchi2000, Scott2004}.
	The AME state with \(n=7\) and \(D=2\) (marked with a cross) is neither excluded by the Scott bound nor by the shadow enumerator, 
	but by Ref.~\cite{Huber2017}.
	The symbol \(\exists\) marks states which are known to exist, constructions can be found in 
	Refs.~\cite{Helwig_graph,Goyeneche2015,Rains1999_nonbinary,Danielsen2012,
	Goyeneche2017,Goyeneche2014, Grassl2004,Danielsen2009,Grassl2015}. 
	In particular, AME states always exist for \(n \leq D\) if \(D\) is a prime-power~\cite{Grassl2004}.
	}
  \label{fig:AME_plot}
\end{figure*}
In this last section, let us return to the question of the existence of 
absolutely maximally entangled (AME) states.
Scott showed in Ref.~\cite{Scott2004} that a necessary requirement for an AME
state of \(n\) parties having \(D\) levels each to exist,
is
\begin{equation}\label{eq:Scott_bound}
  n \leq 
  \begin{cases}
  	2(D^2-1)  &\quad n\,\, \rm{even,} \\
  	2D(D+1)-1 &\quad n\,\, \rm{odd}.
  \end{cases}
\end{equation}
We explain now shortly how this bound was obtained by requiring the positivity 
of the Shor-Laflamme enumerator \(A_{\nhf+2}\).
Recall that complementary reductions of pure states share the same spectrum
and therefore also the same purity.
Thus if \(\ket{\phi_{n,D}}\) is a putative AME state 
of \(n\) parties having \(D\) levels each, then the coefficients of the unitary enumerator as 
defined in Eq.~\eqref{eq:unitary_enum_coeffs} are given by
\begin{equation}\label{AME_purity2}
	A'_k(\ket{\phi_{n,D}}) = \binom{n}{k} D^{- \min(k, n-k)}\,.
\end{equation}
Considering the unitary enumerator coefficient \(A'_{\nhf+2}\), 
only the terms \(A_0=1\), \(A_{\nhf+1}\), and \(A_{\nhf+2}\) contribute,
with appropriate combinatorial prefactors. From Eq.~\eqref{eq:ptrace_sum_A} 
[or from the transform in Eq.~\eqref{eq:unitary_enum_to_KL_enum_poly}], one obtains
\begin{align}\label{AME_scott_deriv1}
 A'_{\nhf+2} &= D^{-(\nhf+2)}\Big[\binom{n}{\nhf+2} A_0 \nn\\
 & + \binom{n-(\nhf+1)}{n-(\nhf+2)} A_{\nhf+1} + A_{\nhf+2} \Big]\,.
\end{align}
The term \(A_{\nhf+1}\) in above equation is fixed by the knowledge of \(A'_{\nhf+1}\),
\begin{equation}\label{AME_scott_deriv2}
 A'_{\nhf+1} = D^{-(\nhf+1)}\Big[\binom{n}{\nhf+1} A_0 + A_{\nhf+1} \Big] \,.
 \end{equation}
Combining Eqs.~\eqref{AME_purity2}, \eqref{AME_scott_deriv1}, and \eqref{AME_scott_deriv2}, 
solving for \(A_{\nhf+2}\), and requiring its non-negativity yields then the bound~of~Eq.~\eqref{eq:Scott_bound}. 
One may wonder if stronger bounds can be obtained by treating the non-negativity of \(A_j\) for \(j>\nhf+2\) in a similar manner. 
However, this does not seem to be the case.

Let us now see what the additional constraints from the shadow enumerator yield.
Having knowledge of all the unitary enumerator coefficients [Eq.~\eqref{AME_purity2}],
all that is left is to evaluate Eq.~\eqref{eq:shadow_in_unitary_enum_Krawtchouk} 
[or Eq.~\eqref{eq:shadow_in_unitary_enum} respectively],
which relates the shadow enumerator to the unitary enumerator.
If any coefficient \(S_j(\ket{\phi_{n,D}})\) happens to be negative, 
a AME state on \(n\) parties having \(D\) levels each cannot exist. 
We should mention that one could also evaluate the shadow inequalities [Eq.~\eqref{eq:shadow_ineq}] 
for a suitable choice of \(T \subseteq \{1\dots n\}\) directly - the shadow coefficients simply represent 
symmetrized forms of these inequalities.
To give an example, consider a putative AME state on four qubits, whose non-existence was proven by Ref.~\cite{Higuchi2000}.
Choosing \(T=\{1,2,3, 4\}\) leads to 
\begin{equation} 
 S_0(\ket{\phi_{4,2}})  = A'_0 - A'_1 + A'_2 - A'_3 + A'_4 
 = -\frac{1}{2} \,\ngeq 0\,,
\end{equation}
in contradiction to the requirement that all \(S_j\) be non-negative.
For general AME states, the coefficient \(S_0\) reads
\begin{align}
 S_0(\ket{\phi_{n,D}}) &= \sum_{k=0}^n (-1)^k \binom{n}{k} D^{-\min(k, n-k)}\,.
\end{align}
The complete set of coefficients is given by [c.f. Eq.~\eqref{eq:shadow_in_unitary_enum_Krawtchouk}]
\begin{align}\label{eq:AME_shadow_coeff}
 S_j(\ket{\phi_{n,D}}) &= \sum_{k=0}^n K_{n-j}(k;n) \binom{n}{k} D^{-\min(k, n-k)} \,,
\end{align}
where \(K_m(k;n)\) is the so-called Krawtchouk polynomial (see Appendix~D).

In Fig.~\ref{fig:AME_plot}, the parameters of hypothetical AME states are shown: 
In dark blue, AME states are marked which are already excluded by the bound from Scott [Eq.~\eqref{eq:Scott_bound}];
in light blue, those AME states are shown for which the negativity of the shadow enumerator coefficients \(S_j(\ket{\phi_{n,D}})\)
[Eq.~\eqref{eq:AME_shadow_coeff}] gives stronger bounds. 
For Fig.~\ref{fig:AME_plot}, all shadow coefficients of hypothetical AME states with local dimension \(D \leq 9\) 
and \(n\) not violating the Scott bound have been evaluated. For \(3\leq D \leq 5\), we found $27$
instances where the shadow enumerator poses a stronger constraint than the bound from Scott.

The non-existence of AME states having parameters \(n=4,9,11\) with \(D=2\) was already known~\cite{Higuchi2000, Scott2004}.
The AME state with \(n=7\) and \(D=2\) (marked with a cross) is neither excluded by the Scott bound nor by the shadow enumerator, 
but by Ref.~\cite{Huber2017}.
The symbol \(\exists\) marks states which are known to exist, constructions can be found in 
Refs.~\cite{Helwig_graph,Goyeneche2015,Rains1999_nonbinary,Danielsen2012,
Goyeneche2017,Goyeneche2014, Grassl2004,Danielsen2009,Grassl2015}. 
In particular, AME states always exist for \(n \leq D\) if \(D\) is a prime-power~\cite{Grassl2004}.

We conclude, that additionally to the known non-existence bounds, 
three-level AME states of \(n = 8,12,13,14,16,17,19,21,23\),
four-level AME states of \(n = 12,16,20,24,25,26,28,29,30,33,37,39\), and 
five-level AME states of \(n = 28,32,36,40,44,48\) parties do not exist.

\section{Mixed-dimensional AME states}
\label{sec:AME_mix}
One might wonder about the existence of absolutely maximally entangled states also in systems that have {\em mixed} local dimensions: 
does a pure state exist, such that every bipartition shows maximal entanglement~\footnote{
For mixed-dimensional states, this differs in definition to the one given by Ref.~\cite{Goyeneche2016}, 
which demands that every reduction of size \(\nhf\) be maximally mixed. 
The choice of definition depends on the desired feature in applications.}? 
For this to be true, every subsystem whose dimension is not larger than that of its complement must be maximally mixed.

Let us give examples of four-partite systems that consist of qubits and qutrits.
As already shown in Ref.~\cite{Higuchi2000}, AME states on four qubits do not exist. On the other hand, an AME state
on four qutrits does exist and is given by a stabilizer state~\cite{Helwig_graph}.
How about other configurations?
Using the shadow inequality, it can be seen that AME states in systems having the local dimensions 
$2\times2\times2\times3$ and $2\times2\times3\times3$ are not allowed.
The last remaining case, a system with dimensions $2\times3\times3\times3$, 
allows for such a state, which we could
find using an iterative semi-definite program (see below) with analytical post-processing.
The state we found reads
 \begin{align}\label{eq:2333ame}
 \ket{\phi_{2333}} = 
&-\a 	\ket{0011} 
-\b  	\ket{0012} 
+\b 	\ket{0021}  
+\a 	\ket{0022} \nn\\
&-\b  \ket{0101}
+\a 	\ket{0102}
+\b 	\ket{0110}
+\a 	\ket{0120} \nn\\
&-\a  	\ket{0201}
+\b \ket{0202}
-\a   	\ket{0210}
-\b   	\ket{0220} \nn\\
&-\b  	\ket{1011}
+\a 	\ket{1012}
-\a  	\ket{1021}
+\b 	\ket{1022} \nn\\
&+\a 	\ket{1101} 
+\b 	\ket{1102} 
-\a  	\ket{1110} 
+\b 	\ket{1120} \nn\\
&- \b 	\ket{1201} 
- \a 	\ket{1202}
- \b 	\ket{1210}
+\a 	\ket{1220} \nn\,.
 \end{align}
 Two possible sets of coefficients are given by
 \begin{align}
  \a &= \frac{1}{6} \sqrt{\frac{3}{2} \pm \frac{\sqrt{65}}{6}}\,, &
  \b &= \frac{1}{54\a} = \frac{1}{6} \sqrt{\frac{3}{2} \mp \frac{\sqrt{65}}{6}}\,,
 \end{align}
 which are (up to a global sign) the two solutions to the constraints
 \begin{align}
    12(\a^2 + \b^2) &= 1\,, 	& 54 \,\a\b &= 1\,.
 \end{align}
 Both solutions are equivalent under local unitaries; the gate \(U~=~\exp(i \varphi \sigma_y) \ot \one^{\ot 3}\)
with
  \(\varphi = 2\arctan(\frac{1-54\a^2}{1+54\a^2}) = -2\arctan \sqrt{5/13}\)
 maps the first to the second solution.
 
We found this state with an iterative semi-definite program that works in the following way~\cite{thesis_FHuber}:
 \begin{itemize}
  
  \item[1)] Choose a random initial state \(\ket{\psi^{(0)}}\).
  
  \item[2)] Solve the following semidefinite program,
    \begin{align}
      \underset{\varrho}{\text{maximize}} 	\quad& \bra{ \psi^{(i)}} \r   \ket{ \psi^{(i)}}   \nonumber\\
      \text{subject to}	 \quad& \r_{AB} = \r_{AC} = \r_{AD} = \one/6 \nn\\
					 & \r_B = \r_C = \r_D = \one/3 \nn\\
					 &\tr[\varrho] = 1, \quad \varrho = \varrho^{\dag}, \quad \varrho \geq 0 \,. \nn
    \end{align}
    (Note that \(\r_{AB} = \one/6\) implies \(\r_A = \one/2\), and latter constraint does not need to be stated separately.)
  
  \item[3)] Set \(\ket{\psi^{(i+1)}}\) equal to the eigenvector corresponding to the maximal eigenvalue of \(\varrho\). 
  
  \item[4)] Repeat steps $2)$ \& $3)$ until convergence.
 \end{itemize}

Thus after each iteration, the state is projected onto the eigenvector corresponding to its largest eigenvalue.
Note that the reductions onto two qutrits need to have rank $6$ in a $9$-dimensional space, with all non-vanishing eigenvalues being equal. 
While this requirement cannot be stated as a semidefinite constraint, the maximal mixedness of the complementary reductions can.
In that way, the above constraints guarantee maximal entanglement across every bipartition. 
Above iterative program may also be used for other configurations and for related problems, 
such as for those presented in Ref.~\cite{Reuvers1711.07943}.

\section{Conclusion}\label{sec:conclusion}
Using the quantum weight enumerator machinery originally derived by Shor, Laflamme and Rains, 
we obtained bounds on the existence of absolutely maximally entangled states, 
excluding 27 open cases of dimensions larger than two.
For this, we used the so-called shadow inequalities, which constrain the possible correlations 
arising from quantum states. Additionally, we provided a proof of the quantum MacWilliams identity 
in the Bloch representation, clarifying its physical interpretation. 
We furthermore raised the question of mixed-dimensional AME states, 
of which it may be possible to form interesting mixed-dimensional QECC~\footnote{
As an example, the state \(\ket{\phi_{2333}}\) can be regarded as a \(((4,1,3))_{2^1 3^3}\) code,
where the lower indices denote the local dimensions.
A partial trace over the last particle yields a \(((3,3,2))_{2^1 3^2}\) code,
which can be used to correct one error if its position is known. This follows from Thm.~19 in Ref.~\cite{Rains1998}.}.

For future work, it would be interesting to see what the generalized shadow inequalities
involving higher-order invariants~\cite{Rains2000} imply for the distribution of correlations 
in QECC and multipartite quantum states.


\begin{acknowledgments}
We thank 
Jens Eisert, 
Markus Grassl,
Christian Majenz, 
Marco Piani,
Ingo Roth,
and
Nikolai Wyderka 
for fruitful discussions. FH thanks Markus Grassl 
for introducing him to the theory of QECC.
We thank the anonymous Referees for suggesting to consider AME states of mixed local dimensions.

This work was supported by 
the Swiss National Science Foundation (Doc.Mobility 165024),
the FQXi Fund (Silicon Valley Community Foundation),
the DFG (EL710/2-1), 
the ERC (Consolidator Grant 683107/TempoQ),
the Basque Country Government (IT986-16), 
MINECO/FEDER/UE (FIS2015-67161-P), 
and the UPV/EHU program UFI 11/55.
\end{acknowledgments}

\section{Appendix}

\subsection{Proof of observation~1}
Let us proof Observation~\ref{eq:ptrace_channel}.
The partial trace can also be written as
\\
\\ \noindent {\bf Observation I.}
The partial trace over subsystem \(S\) tensored by the identity on \(S\)
can also be written as a channel,
\begin{equation}
      \Tr_{S}(M) \ot \one_{S} = D^{-|S|}\sum_{\supp(E) \subseteq S} \!\!\!\! E M E^\dag \,. \\
  \end{equation}

\begin{proof}
 Consider a bipartite system with Hilbert space \( \mc{H} = \C^D \ot \C^D\) with a local orthonormal operator basis \(\{e_j\}\) on \(\C^D\).
 Define the \(\SWAP\) operator as
 \begin{equation}
  \SWAP = \sum_{j,k=0}^{D-1}\ketbra{jk}{kj}\,.
 \end{equation}
 It acts on pure states as \(\SWAP (\ket{\psi} \ot \ket{\phi}) = \ket{\phi} \ot \ket{\psi}\).

  It can also be expressed in terms of any orthonormal basis \(\{e_j\}\)
  as \cite{Wolf2013}
  \begin{equation}
    \SWAP = \frac{1}{D} \sum_{j=0}^{D^2-1} e_j^\dag \ot e_j \,.
  \end{equation}
 Therefore we can express \(\one \ot N \) as 
 \begin{align}
  \one \ot N  
  &=\SWAP \cdot (N \ot \one) \cdot \SWAP \nn\\
  &= (\frac{1}{D}\sum_{ j=0}^{D^2-1} e_j \ot e_j^\dag ) (N \ot \one )( \frac{1}{D}\sum_{ k=0 }^{D^2-1} e_k^\dag \ot e_k ) \nn\\
  &= D^{-2} \sum_{j,k=0}^{D^2-1} (e_j N e_k^\dag) \ot (e_j^\dag e_k) \,.
 \end{align}
 Tracing over the second party gives
  \begin{equation}
    \Tr(N)\one = \frac{1}{D} \sum_{j=0}^{D^2-1} e_j N e_j^\dag \,.
  \end{equation}
 The claim follows from the linearity of the tensor product. This ends the proof. 
\end{proof}
Note that the proof is independent 
of the local orthonormal operator basis \(\{e_j\}\) chosen.


\subsection{Shor-Laflamme enumerators and the distance of QECC}
Here, we prove that for any QECC, the coefficients \(A_j = A_j(\QQ)\) and \(B_j = B_j(\QQ)\) are non-negative
and fulfill
\begin{align}
	K B_0 &= A_0 = K^2      \label{eq:KB0_A0_app}\,,\\
	K B_j &\geq A_j  	\label{eq:KBj_Aj_app}\,.
\end{align}
Furthermore, any projector \(\QQ\) of rank \(K\)
is a QECC of distance \(d\) if and only if \(K B_j(\QQ) = A_j(\QQ)\) for all \(j<d\) .

Let us start with the first claim. The equation \(KB_0 = A_0 = K^2\) follows by direct computation.
Let us show that for a QECC having the parameters \(((n,K,d))_D\),
the coefficients of the Shor-Laflamme enumerator fulfill [cf. Eq.~\eqref{eq:KBj_Aj}]
\(
 A_j(\QQ) \leq K B_j(\QQ)\,,
\)
where equality holds for \(j < d\).
Recall that 
\begin{align}
	A_j(\QQ) &= \sum_{\wt(E) = j} \Tr(E \QQ) \Tr(E^\dag \QQ) \,,\\
	B_j(\QQ) &= \sum_{\wt(E) = j} \Tr(E \QQ E^\dag \QQ)      \,.
\end{align}
Let us check the inequality for each term appearing in the sum, namely for those of the form
\begin{equation}\label{eq:AjBj_term_by_term}
  \Tr( E \QQ ) \Tr (E^\dag \QQ) \leq K \Tr( E \QQ E^\dag \QQ)\,,
\end{equation}
where \(E\) is a specific error under consideration. 
For later convenience, let us choose the error-basis \(\EE\) to be Hermitian, e.g. formed by tensor products of 
the generalized Gell-Mann matrices~\cite{Bertlmann2008}.
Using \(\QQ = \sum_{i=1}^{K} \dyad{i_\QQ}\), write
\begin{align}
  \Tr( E \QQ ) \Tr (E^\dag \QQ) &= (\sum_{i=1}^{K}   \bra{i_\QQ} E \ket{i_\QQ}) (\sum_{j=1}^{K} \bra{j_\QQ} E^\dag \ket{j_\QQ}) \,,\nn\\
  \Tr( E \QQ E^\dag \QQ) 	&=  \sum_{i,j=1}^{K} \bra{i_\QQ} E \dyad{j_\QQ} E^\dag \ket{i_\QQ} \,.
\end{align}
For the case of \(\wt(E) < d\), 
let us recall the definition of a QECC [Eq.~\eqref{eq:def_qecc}],
\begin{equation}
    \bra{i_\QQ} E \ket{j_\QQ} = \delta_{ij} C(E), \quad \text{ if } \wt(E) < d\,.
\end{equation}
This leads for \(j<d\) to
\begin{align}
  \Tr( E \QQ ) \Tr (E^\dag \QQ) &=  K^2 C(E) C^*(E) \,,\nn\\
  \Tr( E \QQ E^\dag \QQ) 	&=  K C(E) C^*(E) \,.
\end{align}
Therefore, \(A_j(\QQ) = K B_j(\QQ)\) for all \(j<d\).

If on the other hand \(\wt(E) \geq d\), 
let us define the matrix \(\CC\) having the entries \(\CC_{ij} = \bra{i_\QQ} E \ket{j_\QQ}\). 
Note that \(\CC\) is a Hermitian matrix of size \(K \times K\). Then
\begin{align}
  \Tr( E \QQ ) \Tr (E^\dag \QQ) &=  [\Tr(\CC)]^2 \,, \nn\\
  \Tr( E \QQ E^\dag \QQ) 	&=  \Tr(\CC^2) \,.
\end{align}
Consider the diagonalization of \(\CC\). By Jensen's inequality, its eigenvalues must fulfill
\begin{equation}\label{eq:AjBj_convexity}
  \Big(\sum_{i=1}^{K} \lambda_i \Big)^2 \leq K \sum_{i=1}^{K} \lambda_i^2  \,,
\end{equation}
from which the inequality \( A_j(\QQ) \leq K B_j(\QQ)\) follows.

Let us now show that a projector \(\QQ\) of rank \(K\) is a QECC of distance \(d\)
if and only if \(A_j(\QQ) = K B_j(\QQ)\) for all \(j < d\). This can be seen in the following way:

``\(\Rightarrow\)'':
Use the definition of QECC, Eq.~\eqref{eq:def_qecc}.

``\(\Leftarrow\)'': 
Note that in order to obtain \(A_j(\QQ) = K B_j(\QQ)\),
there must be equality in Eq.~\eqref{eq:AjBj_term_by_term} for all \(E\) with \(\wt(E)=j\).
Thus, also equality in Eq.~\eqref{eq:AjBj_convexity} is required. 
However, this is only possible if all eigenvalues \(\lambda_i\) of \(\CC\) are equal.
Then, \(\CC\) is diagonal in any basis, and we can write
\begin{equation}
 \bra{i_\QQ} E \ket{j_\QQ} = \delta_{ij} \lambda(\CC) \,.
\end{equation}
Because above equation must hold for all errors \(E\) of weight less than \(d\),
we obtain Eq.~\eqref{eq:def_qecc} defining a quantum error correcting code:
\begin{equation}
  \bra{i_\QQ} E \ket{j_\QQ} = \delta_{ij} C(E) \,,
\end{equation}
for all \(E\) with \(\wt(E) < d\).
This ends the proof.


\subsection{Relating the unitary enumerators to the Shor-Laflamme enumerators}\label{sec:uni_enum_to_KLF_enum}
Let us relate the unitary enumerators to 
the Shor-Laflamme enumerators by means of a polynomial transform.
\begin{align}
	& A'_{MN}(x,y) = \sum_{m=0}^n A'_m(M,N) x^{n-m} y^m \nn\\
		&= \sum_{m=0}^n \Big[ \sum_{j=0}^m \binom{n-j}{n-m} A_j(M,N) D^{-m}\Big] x^{n-m} y^m  \nn\\
		&= \sum_{j=0}^n \sum_{m=0}^n \binom{n-j}{n-m} A_j(M,N) x^{n-m} (y/D)^m  \allowdisplaybreaks\nn\\
		&= \sum_{j=0}^n \sum_{m=j}^n \binom{n-j}{n-m} A_j(M,N) x^{n-m} (y/D)^m (y/D)^{-j} (y/D)^{j} \nn\\
		&= \sum_{j=0}^n \sum_{m=0}^{n-j} \binom{n-j}{n-j-m} A_j(M,N) x^{n-j-m} (y/D)^m (y/D)^{j}  \nn\\
		&= \sum_{j=0}^n A_j(M,N) (x+y/D)^{n-j} (y/D)^{j}  \nn\\
		&= A_{MN} \Big(x+\frac{y}{D}, \frac{y}{D}\Big)\,.
\end{align}
In an analoguous fashion (replace \(A'_m\) by \(B'_m\), and
\(A_j\) by \(B_j\)), one obtains
\begin{equation}
	B'_{MN}(x,y) = B_{MN} \Big(x+\frac{y}{D}, \frac{y}{D} \Big)\,.
\end{equation}


\subsection{Krawtchouk polynomials}

The Krawtchouk (also Kravchuk) polynomials can be seen as 
a generalization of the binomial coefficients.
They are, for \(n,k \in \N_0\) and \(n-k\geq 0\), 
defined as~\footnote{See p. \(42\) in Ref.~\cite{Nebe2006} or Chpt.~$5$, \(\mathsection 7\) in Ref.~\cite{MacWilliams1981}.}
\begin{align}
	K_m(k;n) = \sum_{\alpha} (-1)^\alpha \binom{n-k}{m-\alpha} \binom{k}{\alpha} \,.
\end{align}
If \(m<0\), \(K_m(k;n)=0\).
The generating function of the Krawtchouk polynomial is
\begin{equation}\label{eq:Krawtchouk_genfun}
	\sum_{m} K_m(k;n) z^m = (1+z)^{n-k} (1-z)^k \,.
\end{equation}
In this work, we need a closely related expression,
\begin{equation}\label{eq:Krawtchouk_genfun_ext}
	\sum_{m} K_m(k;n) x^{n-m} y^m =  (x + y)^{n-k} (x-y)^k\,.
\end{equation}
That above equation holds, can be seen in the following way.
\begin{align}
	&(x+ y)^{n-k} (x-y)^k \nn\\
			    &= \sum_\alpha \binom{n-k}{\alpha} x^{n-k-\alpha} (y)^\alpha \,\,
			       \sum_\beta  \binom{k}{\beta}    x^{k-  \beta}  y^\beta (-1)^\beta \nn \allowdisplaybreaks\\
			    &= \sum_\alpha \sum_\beta \binom{n-k}{\alpha} \binom{k}{\beta}  
			              x^{n- (\alpha + \beta)}  y^{(\alpha+\beta)} (-1)^\beta \nn\allowdisplaybreaks\\
			    &= \sum_{m}    \Big[ \sum_\beta \binom{n-k}{m-\beta}\binom{k}{\beta}   (-1)^\beta \Big]
			              x^{n-m}  y^{m} \nn\\
			    &= \sum_{m} K_m(k;n) x^{n-m} y^m \,,
\end{align}
where we set $m=\alpha+\beta$ in the third line. 
Of course, setting $x=1$ recovers Eq.~\eqref{eq:Krawtchouk_genfun}.

We will also need the Krawtchouk-like polynomial 
\begin{align}\label{eq:Krawtchouk_like}
  &\tilde{K}_m(k;n,\gamma,\delta) = \nn\\ &\sum_{\alpha} (-1)^\alpha \binom{n-k}{m-\alpha} \binom{k}{\alpha} \gamma^{[(n-k) - (m-\a)]} \delta^{(m-\alpha)} \,,
\end{align}
which are the coefficients of
\begin{align}
&(\gamma x+ \delta y)^{n-k} (x-y)^k \nn\\
			    &= \sum_\alpha \binom{n-k}{\alpha} (\gamma x)^{n-k-\alpha} (\delta y)^\alpha
			       \sum_\beta  \binom{k}{\beta}    x^{k-  \beta}  y^\beta (-1)^\beta \nn\\
			    &= \sum_\alpha \sum_\beta \binom{n-k}{\alpha} \binom{k}{\beta}  
			              x^{n- (\alpha + \beta)}  y^{(\alpha+\beta)} \gamma^{n-k-\alpha}  \delta^\alpha (-1)^\beta \allowdisplaybreaks \nn\\
			    &= \sum_{m}    \Big[ \sum_\beta \binom{n-k}{m-\beta}\binom{k}{\beta}   (-1)^\beta \gamma^{(n-k)-(m-\b)} \delta^{m-\beta} \Big]
			              x^{n-m}  y^{m} \nn\\
			    &= \sum_{m} \tilde{K}_m(k;n, \gamma, \delta) x^{n-m} y^m \,,
\end{align} 
where we set $m=\alpha+\beta$ in the second last line. 

\subsection{The shadow enumerator in terms of the unitary enumerator}\label{sect:Krawtchouk_polynomials}
Let us now transform the shadow enumerator into the unitary enumerator.
\begin{align}
	S_{M N}(x,y) &= \sum_{m=0}^n  S_m x^{n-m} y^m 				\nn\\
		   &= \sum_{m=0}^n  \sum_{k=0}^n K_{n-m}(k;n) A'_k(M,N) \, x^{n-m} y^m 	  \allowdisplaybreaks\nn\\
		   &= \sum_{k=0}^n  A'_k(M,N) \Big[ \sum_{m=0}^n  K_{n-m}(k;n) \, x^{n-m} y^m \Big]  \allowdisplaybreaks\nn\\
		   &= \sum_{k=0}^n  A'_k(M,N) \Big[ \sum_{m'=0}^n  K_{m'}(k;n) \, x^{m'} y^{n-m'} \Big]	  \allowdisplaybreaks\nn\\
		   &= \sum_{k=0}^n  A'_k(M,N) (y+x)^{n-k} (y-x)^{k}		\nn\\
		   &= A_{MN}'(x+y,y-x)\,.
\end{align}

Above, the second last equality follows from Eq.~\eqref{eq:Krawtchouk_genfun_ext}.

\subsection{Linear programming bound}\label{sect:LP}

For completeness, we provide the linear programming bound that was first established by Refs.~\cite{Shor1997, Rains1999}.
Given the parameters \( ((n,K,d))_D\) of a hypothetical QECC, one can formulate a linear program
to find possible enumerators which satisfy all the relations derived, 
namely Eq.~\eqref{eq:KB0_A0} and Eq.~\eqref{eq:KBj_Aj}, as well as the quantum MacWilliams identity (Thm.~2) 
and the quantum shadow identity (Thm.~3).
If no valid weights \(A_j\) can be found, a code with the proposed parameters cannot exist.
This provides a method to prove the non-existence of certain hypothetical states and QECC; on the other hand, 
the existence of a valid enumerator however does not imply the existence of a corresponding code.

\begin{theorem}[LP bound for general QECC~\footnote{
For \(D=2\), this was stated in Thm.~21 in Ref.~\cite{Calderbank1998}, in Thm.~10 and~12 in Ref.~\cite{Rains1999}, 
and on p.~$383$ in Ref.~\cite{Nebe2006}.}]
 If a \(((n,K,d))_D\) exists, then there is a solution to the following set of linear equations and inequalities:
 \begin{align}\label{eq:LP_for_qecc}
    K B_0 	&= A_0 = K^2 \nn\\
    K B_i 	&= A_i \geq 0	\quad (i <     d)	\nn\\
    K B_i 	&\geq A_i 	\geq 0 \quad (i \geq d) 	\nn\\
    B_i 	&= {D^{-n}} \sum_{0\leq k \leq n} \tilde{K}_i(k;n, 1, D^2-1) A_k \nn\\
    S_i 	&= {D^{-n}} \sum_{0\leq k \leq n} (-1)^k \tilde{K}_i(k;n, D-1, D+1) A_k \nn\\
    S_i 	&\geq 0 \,,
 \end{align}
 where Krawtchouk-like polynomial \(\tilde{K}_i(k;n,\gamma, \delta)\) is given by Eq.~\eqref{eq:Krawtchouk_like}.
 For pure codes, the second constraint above is strengthened to the equality
 \begin{equation}
    K B_i 	= A_i = 0  \quad (i <     d)\,.
 \end{equation}
 For qubit stabilizer codes, one additionally has that either one of the two below conditions is satisfied
 \begin{equation}
  \sum_{i \text{ even}} A_i = 
    \begin{cases}
  2^{n-\log_2 (K)-1}    	\quad\text{(type \(I\) codes)}\\
  2^{n-\log_2 (K)} 		\quad \quad\text{(type \(II\) codes)}\,.
  \end{cases}
 \end{equation}
For self-dual codes (where $K=1$), additionally \(S_{n-j}=0\) holds for all odd \(j\).
\end{theorem}

\onecolumngrid
\subsection{An overview on the identities between the weight enumerator polynomials.}
In below table, we summarize the known relations between the weight enumerator polynomials.
\vspace{1.2em}
\begin{table*}[h!]
\begin{tabular}{ l  l  l }
\multirow{3}{*}{Shor-Laflamme and unitary enum.:\quad\quad}
&$A'_{MN}(x,y) = A_{MN}\big(x+\frac{y}{D},\frac{y}{D}\big)$ & $A_{MN}(x,y) = A'_{MN}(x-y,Dy )$\\
&$B'_{MN}(x,y) = B_{MN}\big(x+\frac{y}{D},\frac{y}{D}\big)$ & $B_{MN}(x,y) = B'_{MN}(x-y,Dy )$\\
&$A'_{MN}(x,y) = B'_{MN}(y,x)$ &\\
\noalign{\vskip 2mm}    
\hline
\noalign{\vskip 2mm}
\multirow{1}{*}{MacWilliams identity.:\quad\quad}
      & $A_{MN}(x,y) = B_{MN} \Big( \frac{x+(D^2-1)y}{D}, \frac{x-y}{D} \Big)$ \quad\quad 
      & $B_{MN}(x,y) = A_{MN} \Big( \frac{x+(D^2-1)y}{D}, \frac{x-y}{D} \Big)$ \\
\noalign{\vskip 2mm}
\hline
\noalign{\vskip 2mm}    
  \multirow{2}{*}{Shadow identity.:\quad\quad}
      &$S_{MN}(x,y) = A_{MN}\left( \frac{ (D-1)x + (D+1)y}{D}, \frac{y-x}{D} \right) $ &\\
      &$S_{M N}(x,y) = A_{MN}'(x+y,y-x)$ &
\end{tabular}
\caption{An overview on the identities between the weight enumerator polynomials.} 
\label{table:identities}
\end{table*}
\twocolumngrid


\begin{thebibliography}{99}
  
\bibitem{Helwig2012}
  W. Helwig, W. Cui, J. I. Latorre, A. Riera, and H.-K. Lo,
  Phys. Rev. A {\bf 86}, 052335 (2012).  

\bibitem{Helwig2013}
  W. Helwig and W. Cui,
  arXiv:1306.2536. 

\bibitem{Latorre2015}
  J.~I.~Latorre, G.~Sierra,
  arXiv:1502.06618.

 \bibitem{Pastawski2015}
  F.~Pastawski, B.~Yoshida, D.~Harlow, J.~Preskill,
  JHEP {\bf 06}, 149 (2015).
  
\bibitem{Li2017}
   Y.~Li, M.~Han, M.~Grassl, and B.~Zeng,
  New J. Phys. {\bf 19}, 063029 (2017).

\bibitem{Scott2004}
  A.~J.~Scott, 
  Phys. Rev. A {\bf 69}, 052330 (2004).

\bibitem{Gisin1998}
  N. Gisin and H. Bechmann-Pasquinucci,
  Phys. Lett. A {\bf 246}, 1  (1998).
   
\bibitem{Higuchi2000}
  A. Higuchi and A. Sudbery,
  Phys. Lett. A {\bf 273}, 213 (2000).

\bibitem{Huber2017}
  F.~Huber, O.~G\"uhne, and J.~Siewert,
  Phys. Rev. Lett. {\bf 118}, 200502 (2017).

\bibitem{Hein2006}
  M.~Hein, J.~Eisert, and H.~J.~Briegel,
  Phys. Rev. A {\bf 69}, 062311 (2004).

\bibitem{Helwig_graph}
  W. Helwig, 
  arXiv:1306.2879.

\bibitem{Grassl2004}
  M.~Grassl, T.~Beth, M.~Roetteler,
   Int. J. Quantum Inform. {\bf 2}, 55 (2004).

\bibitem{Goyeneche2015}
  D. Goyeneche, D. Alsina, J. I. Latorre, A. Riera, and K. {\.Z}yczkowski, 
  Phys. Rev. A {\bf 92}, 032316 (2015).

\bibitem{Goyeneche2017}
  D.~Goyeneche, Z.~Raissi, S.~D.~Martino, K.~{\.Z}yczkowski,
  arXiv:1306.2879.

\bibitem{Shor1997}
  P.~Shor, R.~Laflamme,
  Phys. Rev. Lett. {\bf 78}, 1600 (1997).

\bibitem{Rains1998}
  E.~M.~Rains,
  IEEE Trans. Inf. Theory {\bf 44}, 1388 (1998).

\bibitem{Rains1999}
  E.~M.~Rains,
  IEEE Trans. Inf. Theory {\bf 45}, 2361 (1999).
  
\bibitem{Rains2000}
  E.~M.~Rains,
  IEEE Trans. Inf. Theory {\bf 46}, 2361 (2000).
  
\bibitem{Bertlmann2008}
  R.~A.~Bertlmann and P.~Krammer,
  J.~Phys.~A: Math. Theor. {\bf 41}, 235303 (2008).
  
\bibitem{Gross2007}
  D.~Gross, K.~Audenaert, and J.~Eisert,
  J. Math. Phys. {\bf 48}, 052104 (2007).

\bibitem{Gottesmann2002}
  D.~Gottesman,
  in Quantum Computation: A Grand Mathematical Challenge for the Twenty-First Century and the Millennium, 
  pp. 221-235 (American Mathematical Society, Providence, Rhode Island, 2002).

\bibitem{Nebe2006}
  G.~Nebe, E.~M.~Rains, and N.~J.~A.~Sloane,
  {\it Self-Dual Codes and Invariant Theory},
  Springer (Berlin Heidelberg) (2006).

\bibitem{Calderbank1998}
  A.~R.~Calderbank, E.~M.~Rains, P.~W.~Shor, and N.~J.~A.~Sloane,
  IEEE Trans. Inf. Theory  {\bf 44}, 1369 (1998).

\bibitem{Aschauer2004}
  H.~Aschauer, J.~Calsamiglia, M.~Hein, and H.~J.~Briegel,
  Quant. Inf. Comp. {\bf 4}, 383 (2004).
  
\bibitem{Kloeckl2015}
  C.~Klöckl and M.~Huber,
  Phys. Rev. A {\bf 91}, 042339 (2015).

\bibitem{Gottesmann1997}
  D.~Gottesmann,
  PhD thesis, quant-ph/9705052.
 
\bibitem{Rains1999_nonbinary}
  E.~M.~Rains,
  IEEE Trans. Inf. Theory {\bf 45}, 1827 (1999).  

\bibitem{Danielsen2012}
  L.~E.~Danielsen,
  IEEE Trans. Inf. Theory {\bf 58}, 5500 (2012).

\bibitem{Goyeneche2014}
  D.~Goyeneche, K.~{\.Z}yczkowski
  Phys. Rev. A {\bf 90}, 022316 (2014)

\bibitem{Danielsen2009}
  L.~E.~Danielsen,
  Adv. Math. Commun. {\bf 3}, 329 (2009).

\bibitem{Grassl2015}
  M.~Grassl, M.~Roetteler,
  2015 Proc. IEEE Int. Symp. Inf. Theory (ISIT), Hong Kong, 1104-1108 (2015).
  
 \bibitem{Wolf2013}
  M.~M.~Wolf,
  lecture notes:
  {Quantum Channels \& Operations, a guided tour} (2012), 
  available online at \url{http://www-m5.ma.tum.de/foswiki/pub/M5/Allgemeines/MichaelWolf/QChannelLecture.pdf}.

\bibitem{Goyeneche2016}
D.~Goyeneche, J.~Bielawski, and K.~{\.Z}yczkowski, 
Phys. Rev. A {\bf 94}, 012346 (2016).
  
\bibitem{thesis_FHuber}
  F.~Huber,  PhD Thesis, 
  ``Quantum States and their Marginals: 
  From Multipartite Entanglement to Quantum Error-Correcting Codes'' (2017), submitted.

\bibitem{Reuvers1711.07943}
R.~Reuvers,
arXiv:1711.07943.

\bibitem{MacWilliams1981}
  F.~J.~MacWilliams and N.~J.~A.~Sloane,
  The Theory of Error-Correcting Codes,
  North-Holland (1981).
 
\end{thebibliography}
\end{document}